\documentclass[letterpaper,11pt]{article} 

\usepackage{amsmath, amsfonts, amssymb, amsthm, amscd, graphicx, 
enumerate, footmisc, mathtools, xcolor}
\providecommand{\U}[1]{\protect\rule{.1in}{.1in}}
\def\theenumi{\arabic{enumi}}

\def\theenumii{\alph{enumii}}
\def\p@enumii{\theenumi.}

\def\theenumiii{\arabic{enumiii}}
\def\p@enumiii{(\theenumi)(\theenumii)}

\def\p@enumiv{\p@enumiii.\theenumiii}
\usepackage{fullpage}
\usepackage{mathrsfs}
\newtheorem{theorem}{Theorem}[section]

\newtheorem{definition}[theorem]{Definition}

\numberwithin{equation}{section}

\usepackage{titling}
\usepackage{enumitem} 
\setlist{noitemsep,topsep=0pt,parsep=0pt} 
\usepackage[margin=1cm]{caption} 
\usepackage{subfig}
\usepackage{tikz}
\usetikzlibrary{calc,fit}

\numberwithin{equation}{section}
\usepackage[bookmarks=true,hypertexnames=false]{hyperref}
\hypersetup{colorlinks=true, citecolor=blue, linkcolor=red, urlcolor=blue}
\makeatletter

\newcommand{\Rmnum}[1]{\expandafter\@slowromancap\romannumeral #1@}
\makeatother

\newcommand{\Holant}{\operatorname{Holant}}
\newcommand{\PlHolant}{\operatorname{Pl-Holant}}

\newenvironment{remark}{\medskip{\bfseries \noindent Remark:}}{\par\medskip}{\par\medskip}

\usepackage[export]{adjustbox} 
\usepackage{pgfplots}
\pgfplotsset{compat=newest}
\def\borderColor{blue!60}
\def\scale{0.6}
\def\nodeDist{1.4cm}
\tikzstyle{internal} = [draw, fill, shape=circle]
\tikzstyle{external} = [shape=circle]
\tikzstyle{square}   = [draw, fill, rectangle, inner sep=5pt]
\tikzstyle{pentagon} = [draw, fill, regular polygon, regular polygon sides=5, inner sep=2pt, minimum size=14pt]

\pgfdeclarelayer{background}
\pgfsetlayers{background,main}

\usepackage{caption}
\usepackage{subfig}
\newcommand{\AllDistinct}{\textsc{All-Distinct}}

\newcommand{\SHARPP}{{\#\rm{P}}}
\newcommand{\AD}{\operatorname{AD}}

\begin{document}

\title{{\bf The Complexity of Counting Edge Colorings for Simple Graphs}}

\vspace{0.3in}

\author{Jin-Yi Cai\thanks{Department of Computer Sciences, University of Wisconsin-Madison. Supported by NSF CCF-1714275.} \\ {\tt jyc@cs.wisc.edu}
\and Artem Govorov\thanks{Department of Computer Sciences, University of Wisconsin-Madison. Supported by NSF CCF-1714275. 
 } \thanks{Artem Govorov is the author's preferred spelling 
of his name,
rather than the official spelling Artsiom Hovarau.
 } \\ {\tt hovarau@cs.wisc.edu}}

\date{}
\maketitle

\bibliographystyle{plainurl}

\begin{abstract}
We prove \#P-completeness results for counting edge colorings on \emph{simple} graphs.
These strengthen the corresponding results on multigraphs from \cite{cgw-focs-2014}.
We prove that for any $\kappa \ge r \ge 3$ counting $\kappa$-edge colorings
on $r$-regular simple graphs is \#P-complete.
Furthermore, we show that for planar $r$-regular simple graphs where $r \in \{ 3, 4, 5 \}$
counting edge colorings with $\kappa$ colors for any $\kappa \ge r$ is also \#P-complete.
As there are no planar $r$-regular simple graphs for any $r > 5$,
these statements cover all interesting cases in terms of the 
parameters  $(\kappa, r)$.
\end{abstract}

\thispagestyle{empty}
\clearpage
\setcounter{page}{1}

\section{Introduction}


A proper edge $\kappa$-coloring, or simply
an edge $\kappa$-coloring, of a graph $G$ is
a labeling of its edges using at most $\kappa$ distinct symbols, 
called \emph{colors},  such that any two incident edges have different colors.
The number of (proper) edge $\kappa$-colorings of $G$ is an important
graph parameter, and can be naturally expressed
as the Holant value on the graph  $G$ where  every vertex of $G$ 
is assigned the  {\sc All-Distinct} constraint function.

%

Given a graph  $G$, it is an interesting problem to determine
 how many colors are required to properly edge color $G$.
The minimum number needed is called the
edge chromatic number, or chromatic index, and is denoted by $\chi'(G)$.
An obvious lower bound is $\chi'(G) \ge \Delta(G)$,
the maximum degree of the graph.
Vizing's Theorem~\cite{Viz65,Diestel-book} states that 
$\Delta(G) + 1$ colors always suffice for simple graphs (i.e. graphs without self-loops or parallel edges).
Whether $\Delta(G)$ colors suffice depends on the graph $G$, and is
clearly a problem in NP.

Consider the edge coloring problem over $3$-regular simple graphs $G$.
It is well known that
if a $3$-regular graph $G$ contains a bridge (i.e., a cut consisting of one edge)
 then it has no edge $3$-coloring.
This follows from a parity argument;
see \cite{Isaacs-1975}
and also a simple proof in section 2 of \cite{Hol81}.
For bridgeless planar $3$-regular simple graphs,
Tait~\cite{Tai80} proved in 1880 that
an edge $3$-coloring always exists iff the Four Color Conjecture (now
Theorem) holds. 
Thus, to the existence question
of an edge $3$-coloring on a planar $3$-regular simple graph,
the answer is 
that there is an edge $3$-coloring iff the graph is bridgeless.
It follows that edge $3$-colorability for
$3$-regular graphs, as a problem  in NP, when restricted
to planar simple  graphs  is solvable in P by the trivial algorithm of
checking that the graph is bridgeless.
 However note that
the correctness of this algorithm is certainly not trivial:
by Tait's Theorem, this correctness is equivalent to a proof of
the Four Color Theorem.

Without the planarity restriction,
Holyer~\cite{Hol81} proved that
determining if a $3$-regular simple graph has an edge $3$-coloring 
is NP-complete.
Leven and Galil~\cite{LG83} extended this result to all
$\kappa \ge 3$: edge $\kappa$-colorability over $\kappa$-regular 
simple graphs is NP-complete for all $\kappa \ge 3$.

We say a polynomial time reduction $f$ from {\sc Sat} to
a decision problem $\Pi$
is parsimonious if it preserves the number of solutions,
i.e., for any instance $\Phi$ of {\sc Sat},
the number of solutions to the problem $\Pi$ on the instance
$f(\Phi)$ is the same as the number of satisfying assignments
of $\Phi$.
A parsimonious
reduction from {\sc Sat} to $\Pi$  also implies that
 the counting version of the problem $\Pi$
is \#P-complete.
Most NP-completeness reductions from {\sc Sat} to decision problems
are, or can be made to be, parsimonious~\cite{Berman-Hartmanis,JanosSimon}.



However,
the reductions by Holyer~\cite{Hol81} and by Leven-Galil~\cite{LG83}
 are not parsimonious.
In fact  we have the following result from~\cite{Wel93}
(p.~118, attributed to an unpublished work by
Edwards and Welsh~\cite{Edwards-Welsh}) that for any $\kappa \ge 4$,
 no parsimonious reduction exists  from  {\sc Sat} to edge $\kappa$-coloring,
unless P $=$ NP, which we explain below using a reduction due
to Blass and Gurevich~\cite{Blass-Gurevich} with a slight modification.


To discuss the (non)-existence of parsimonious reductions 
from  {\sc Sat},
we define a solution to edge $\kappa$-coloring
as a partition of the edge set into $\kappa$ pairwise disjoint matchings
(some of which may be empty).
For graphs $G$ with maximum degree $\Delta(G) =\kappa$,
if  $\chi'(G) = \kappa$, then all
 $\kappa$ matchings must be nonempty.
A theorem due to Thomason~\cite{Thomason1978} says that, for $\kappa \ge 4$,
among graphs with $\chi'(G) = \kappa$ and ignoring isolated vertices,
the only graph that has a unique (proper) edge $\kappa$-coloring is the 
star $K_{1,\kappa}$ (uniqueness in the sense of
counting partitions).
Without the assumption  $\chi'(G) = \kappa$,
and again ignoring isolated vertices, here is a complete list
of graphs having a unique edge $\kappa$-coloring
(again uniqueness in the sense of
counting partitions):
Clearly $\chi'(G) \le \kappa$. Suppose $\chi'(G) < \kappa$.
Let $\sigma$ be a proper $\kappa$-coloring using $\chi'(G)$ colors,
 and let
$C_\alpha$ be the set of $\alpha$-colored edges. Then every  non-empty
$C_\alpha$ consists of a single edge, for otherwise
one can split it by an unused color from $[\kappa]$
and violate the uniqueness. Moreover, for any $\alpha\not =
\beta$, if $C_\alpha$ and $C_\beta$
are non-empty  then the two edges are incident, for otherwise
we can recolor the edge in $C_\beta$  by $\alpha$ and violate the uniqueness.
It follows that, in addition to $K_{1,\kappa}$, the graph $G$
can only be $C_3$ (a cycle of length 3) or among 
$\{K_{1, j} \mid 1 \le j < \kappa\}$ or just the empty graph, and indeed these graphs
have a unique  edge $\kappa$-coloring according to the definition by counting
partitions.
Obviously it is linear time verifiable whether a graph $G$ is
one of $K_{1, j}$ ($1 \le j \le \kappa$) or $C_3$ or the empty graph,
with a finite union of isolated vertices.

Now suppose for a contradiction that a parsimonious reduction $f$ exists
from  {\sc Sat} to edge $\kappa$-coloring, for some  $\kappa \ge 4$.
Let $\Phi$ be an instance of {\sc Sat} in conjunctive normal form,
$\Phi = C_1 \wedge C_2 \wedge \ldots \wedge C_m$,
where each $C_i$ is a disjunction of literals on the variables
$\{x_1, x_2, \ldots, x_n\}$. Let $y$ be a new variable.
Define another Boolean formula
\[\Phi' = [\neg y \vee C_1]  \wedge [\neg y \vee C_2]  \wedge \ldots \wedge 
[\neg y \vee C_m]  \wedge [y \vee x_1]  \wedge [ y \vee x_2]
 \wedge \ldots \wedge  [ y \vee x_n].\]
Clearly $\Phi'$ has the following satisfying assignments,
$y=0, x_1 = x_2 = \ldots = x_n =1$,
together with $y=1$ and all satisfying assignments $(x_1, x_2, \ldots, x_n)$
to $\Phi$.  Therefore  $\Phi'$ has exactly one more 
satisfying assignment than $\Phi$ does. 
It follows that if $\Phi \not \in$ {\sc Sat},
then  $\Phi'$ has a unique satisfying assignment,
which implies that $f(\Phi')$ has a unique solution for edge $\kappa$-coloring;
if $\Phi \in$  {\sc Sat},
then  $\Phi'$ has more than one satisfying assignments,
which implies that
  $f(\Phi')$ does not have a unique solution for edge $\kappa$-coloring.
But uniqueness of solutions to edge $\kappa$-coloring
is linear time testable by Thomason's theorem. This would
imply P $=$ NP.

\begin{remark}
We remark that the notion of \emph{unique} edge $\kappa$-coloring
in the sense of counting partitions is the standard one;
by contrast if we are to count the number of coloring assignments,
we would incur a factor $\kappa!$ for graphs with $\chi'(G) =\kappa$, and
there would be no parsimonious reduction for a trivial reason.
In the rest of the paper, we will adopt the notion of
counting edge colorings by counting the number of
proper edge coloring assignments.
For $\Delta(G) = \kappa$, a fortiori for $\kappa$-regular graphs,
if $\chi'(G) =\kappa$,  then the number of
solutions according to the two definitions differ by exactly a factor
$\kappa!$. Thus  our \#P-hardness results in Theorems~\ref{thm:edge_coloring-regular-simple-non-planar}
and \ref{thm:edge_coloring-regular-simple-planar}
also imply that counting edge $\kappa$-coloring is \#P-complete
in the sense of counting partitions, by simply restricting to 
$\kappa$-regular graphs.
\end{remark}

It is reasonable to suppose that  this lack of
parsimonious reductions for edge coloring
is the reason that it remained an open problem for some decades
whether the counting problem for edge coloring is \#P-complete,
until it was proved in~\cite{cgw-focs-2014}.
Their proof is carried out as part of a classification
program on counting problems, in particular a certain subclass of
Holant problems on higher domain sizes.  However, as is
generally true for Holant problems, the \#P-hardness applies to
multigraphs, i.e., parallel edges and loops are allowed.
By contrast, a simple graph does not have  parallel edges or loops.
Our main purpose in this paper is to prove:

\begin{theorem} \label{thm:edge_coloring-regular-simple-non-planar}
 \#$\kappa$-\textsc{EdgeColoring} is $\SHARPP$-complete
 over $r$-regular simple graphs for any integers $\kappa \ge r \ge 3$.
\end{theorem}

\begin{theorem} \label{thm:edge_coloring-regular-simple-planar}
 \#$\kappa$-\textsc{EdgeColoring} is $\SHARPP$-complete
 over planar $r$-regular simple graphs for any integers $\kappa \ge r$ and $r \in \{3, 4, 5\}$.
\end{theorem}

Note that for $r > 5$, simple planar $r$-regular graphs
 \emph{do not} exist, by Euler's formula. See
 Theorem~\ref{thm:edge_coloring:k=r=3_simple}.

Because these are hardness results, the weaker statement also
holds without the restriction on regularity, or on planarity,
thus for any $\kappa \ge \Delta \ge 3$,
\#$\kappa$-\textsc{EdgeColoring} is $\SHARPP$-complete
 for simple graphs with maximum degree  $\Delta$.
Obviously if $\kappa < \Delta$ there are no proper edge $\kappa$-colorings.
For $\Delta \le 2$, the problem can be easily solved in polynomial time,
since for maximum degree at most 2,
 a graph is a disjoint union of cycles and paths.
Thus Theorems~\ref{thm:edge_coloring-regular-simple-non-planar}
and \ref{thm:edge_coloring-regular-simple-planar} settle the complexity of counting
edge colorings for simple graphs
for all parameter settings of $(\kappa, r)$ or $(\kappa, \Delta)$.
%

Our proof of   Theorems~\ref{thm:edge_coloring-regular-simple-non-planar}
and \ref{thm:edge_coloring-regular-simple-planar}
use reductions from the \#P-hardness for edge colorings on
multigraphs. At the heart of the reductions are some geometrically
inspired gadget constructions using the Platonic solids, such as the
icosahedron (see Figure~\ref{fig:icosahedron}).

%
%
%
%
%
%
%

\section{Preliminaries}

Counting $\kappa$-edge colorings can be naturally expressed as 
a Holant problem on a domain of size $\kappa$.
Let us first recap
the framework of Holant problems on the domain $[\kappa] = \{1, 2, \ldots, \kappa\}$.


Fix a positive integer $\kappa \ge 1$.
A constraint function, or  \emph{signature}, $f$ has an arity $n \ge 0$,
and  is a mapping  from
 $[\kappa]^n \rightarrow \mathbb{C}$.
When $n = 0$, we think of $f$ as a scalar.
For the reason of Turing computability we assume all signatures take
complex algebraic values.


A Holant problem $\Holant(\mathcal{F})$ is parameterized by a set of
signatures $\mathcal{F}$.
An input to the problem $\Holant(\mathcal{F})$ is a
 \emph{signature grid} $\Omega = (G, \pi)$ 
consisting of  a (multi)graph $G = (V,E)$,
where $\pi$ assigns to each vertex $v \in V$ some $f_v \in \mathcal{F}$ 
of arity $\deg(v)$,
 and assigns its incident edges 
to its input variables.
\begin{definition}
 Given a set of signatures $\mathcal{F}$,
 we define the counting problem $\Holant(\mathcal{F})$ as:

 Input: A \emph{signature grid} $\Omega = (G, \pi)$;

 Output: $\Holant_\Omega = \sum_{\sigma} \prod_{v \in V} f_v(\sigma \mid_{E(v)})$, where the sum is  over all edge assignments $\sigma: E \to [\kappa]$,
 $E(v)$ denotes the incident edges of $v$ and $\sigma \mid_{E(v)}$ denotes the restriction of $\sigma$ to $E(v)$.
\end{definition}


By default $G$ is a multigraph so it  may have self-loops or parallel edges.
A \emph{simple} graph does not have self-loops or parallel edges.
We may restrict the  Holant problem to
signature grids where input graphs $G$ are simple graphs, or simple planar graphs.
We use  $\PlHolant(\mathcal{F})$ to denote
the problem restricted to planar signature grids.

%

A function $f_v$ can be represented by listing its values according to the lexicographical order of the input tuples 
which is a vector in $\mathbb{C}^{\kappa^{\deg(v)}}$.
A signature is called \emph{symmetric} if it is invariant under
a permutation of its variables. 
An example is the \textsc{Equality} signature $=_r$ of arity $r$
(on domain $[\kappa]$), which outputs value $1$ if all $r$ input values are equal,
and $0$ otherwise.
Here $=_0$ is the scalar $1$;
$=_1$ is the univariate function that takes constant value $1$.
A binary signature, i.e., a signature of arity 2, can 
 be represented by a
$\kappa \times \kappa$  signature matrix.
E.g.,  the binary \textsc{Equality} signature $=_2$
is represented by the identity matrix $I_{\kappa}$.

%
%

Replacing a signature $f \in \mathcal{F}$ by a constant multiple $c f$,
where $c \ne 0$,
does not change the complexity of $\Holant(\mathcal{F})$.
It introduces a global nonzero factor to $\Holant_\Omega$.

We allow $\mathcal{F}$ to be an infinite set.
We say $\Holant(\mathcal{F})$ is (P-time) tractable,
if $\Holant_\Omega$
is  computable in polynomial time 
where the input description of  $\Omega$ includes the signatures used
 in $\Omega$.
We say $\Holant(\mathcal{F})$ is $\SHARPP$-hard if there exists a finite subset of $\mathcal{F}$ for which the problem is $\SHARPP$-hard.
The same definitions apply for $\PlHolant(\mathcal{F})$ when $\Omega$ is
restricted to planar signature grids.
%
We use  $\le_T$ and  $\equiv_T$  to denote
 polynomial time Turing reduction and equivalence, respectively.


We say a signature $f$ is \emph{realizable} or \emph{constructible} from a signature set $\mathcal{F}$
if there is a graph fragment, called a gadget,
which is a graph with some dangling edges such that each vertex is assigned a signature from $\mathcal{F}$,
and the resulting graph defines $f$ by the Holant sum 
with inputs on the dangling edges.
If $f$ is realizable from a set $\mathcal{F}$,
then we can freely add $f$ into $\mathcal{F}$ while preserving the complexity.
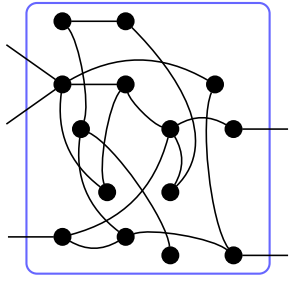
\begin{figure}[t]
 \centering
 \begin{tikzpicture}[scale=\scale,transform shape,node distance=\nodeDist,semithick]
  \node[external]  (0)                     {};
  \node[internal]  (1) [below right of=0]  {};
  \node[external]  (2) [below left  of=1]  {};
  \node[internal]  (3) [above       of=1]  {};
  \node[internal]  (4) [right       of=3]  {};
  \node[internal]  (5) [below       of=4]  {};
  \node[internal]  (6) [below right of=5]  {};
  \node[internal]  (7) [right       of=6]  {};
  \node[internal]  (8) [below       of=6]  {};
  \node[internal]  (9) [below       of=8]  {};
  \node[internal] (10) [right       of=9]  {};
  \node[internal] (11) [above right of=6]  {};
  \node[internal] (12) [below left  of=8]  {};
  \node[internal] (13) [left        of=8]  {};
  \node[internal] (14) [below left  of=13] {};
  \node[external] (15) [left        of=14] {};
  \node[internal] (16) [below left  of=5]  {};
  \path let
         \p1 = (15),
         \p2 = (0)
        in
         node[external] (17) at (\x1, \y2) {};
  \path let
         \p1 = (15),
         \p2 = (2)
        in
         node[external] (18) at (\x1, \y2) {};
  \node[external] (19) [right of=7]  {};
  \node[external] (20) [right of=10] {};
  \path (1) edge                             (5)
            edge[bend left]                 (11)
            edge[bend right]                (13)
            edge node[near start] (e1) {}   (17)
            edge node[near start] (e2) {}   (18)
        (3) edge                             (4)
        (4) edge[out=-45,in=45]              (8)
        (5) edge[bend right, looseness=0.5] (13)
            edge[bend right, looseness=0.5]  (6)
        (6) edge[bend left]                  (8)
            edge[bend left]                  (7)
            edge[bend left]                 (14)
        (7) edge node[near start] (e3) {}   (19)
       (10) edge[bend right, looseness=0.5] (12)
            edge[bend left,  looseness=0.5] (11)
            edge node[near start] (e4) {}   (20)
       (12) edge[bend left]                 (16)
       (14) edge node[near start] (e5) {}   (15)
            edge[bend right]                (12)
       (16) edge[bend left,  looseness=0.5]  (9)
            edge[bend right, looseness=0.5]  (3);
  \begin{pgfonlayer}{background}
   \node[draw=\borderColor,thick,rounded corners,fit = (3) (4) (9) (e1) (e2) (e3) (e4) (e5), transform shape=false] {};
  \end{pgfonlayer}
 \end{tikzpicture}
 \caption{\label{fig:Fgate}An $\mathcal{F}$-gate with 5 dangling edges}
\end{figure}
Formally,
such a notion is defined by an $\mathcal{F}$-gate~\cite{cai-chen-book}. 
An $\mathcal{F}$-gate is similar to a signature grid $(G, \pi)$ for $\Holant(\mathcal{F})$ except that $G = (V,E,D)$ is a graph with some dangling edges $D$.
The dangling edges define external variables for the $\mathcal{F}$-gate.
(See Figure~\ref{fig:Fgate} for an example.)
We denote the regular edges in $E$ by $1, 2, \dotsc, m$ and the dangling edges in $D$ by $m+1, \dotsc, m+n$.
Then we can define a function $\Gamma$ for this $\mathcal{F}$-gate as
\[
 \Gamma(y_1, \dotsc, y_n) = \sum_{x_1, \dotsc, x_m \in [\kappa]} H(x_1, \dotsc, x_m, y_1, \dotsc, y_n),
\]
where $(y_1, \dotsc, y_n) \in [\kappa]^n$ is an assignment on the dangling edges
and $H(x_1, \dotsc, x_m, y_1, \dotsc, y_n)$ is the value of the signature grid on an assignment of all edges in $G$,
which is the product of evaluations at all internal vertices.
We also call this function $\Gamma$ the signature of the $\mathcal{F}$-gate,
or informally, of the gadget. 

An $\mathcal{F}$-gate is planar if the underlying graph $G$ is a planar graph,
and the dangling edges,
ordered counterclockwise corresponding to the order of the input variables,
are in the outer face in a planar embedding.
A (planar) $\mathcal{F}$-gate can be used in a (planar) signature grid as if it is just a single vertex with the particular signature.
An $\mathcal F$-gate is simple if the underlying graph is simple.

Using the idea of planar $\mathcal{F}$-gates,
we can reduce one planar Holant problem to another.
Suppose $g$ is the signature of some planar $\mathcal{F}$-gate.
Then $\PlHolant(\mathcal{F} \cup \{g\}) \leq_T \PlHolant(\mathcal{F})$.
The reduction is simple.
Given an instance of $\PlHolant(\mathcal{F} \cup \{g\})$,
by replacing every appearance of $g$ by the $\mathcal{F}$-gate,
we get an instance of $\PlHolant(\mathcal{F})$.
Since the signature of the $\mathcal{F}$-gate is $g$,
the Holant values for these two signature grids are identical.




An arity $r$ signature on domain size $\kappa$ is fully specified by $\kappa^r$ values.
A \emph{symmetric signature} of arity $r$ on domain $[\kappa]$
can be specified by $\binom{r + \kappa -1}{\kappa -1}$ values, 
namely the signature value is determined by how many variables ($r_1 \ge 0$)
take the first value in $[\kappa]$, 
and how many variables ($r_2 \ge 0$)
take the second value in $[\kappa]$, etc., such that
$r_1 + r_2 + \ldots  + r_{\kappa} = r$. 

The signature that defines the problem of edge coloring
is the {\sc All-Distinct} function.
The signature $\AllDistinct_{r,\kappa}$, or $\AD_{r,\kappa}$ for short, has 
arity $r$ on domain size $\kappa$. It outputs value~$1$ 
when all inputs are distinct and~$0$ otherwise.
Here $\AD_0$ is the scalar $1$;
$\AD_1$ is the univariate function that takes constant value $1$.
%
When $\mathcal F = \{ \AD_{r,\kappa} \mid r \ge 0 \}$,
the problem $\Holant(\mathcal F)$ is called $\#\kappa$-\textsc{EdgeColoring}.

We prove our \#P-completeness results for  counting
edge colorings over (regular) simple graphs by
reducing from the corresponding results for multigraphs 
(Theorem 1.1 in \cite{cgw-focs-2014}).
We restate it below for easy reference.
\begin{theorem}\label{thm:edge_coloring}
$\#\kappa$-\textsc{EdgeColoring} is $\SHARPP$-complete over planar $r$-regular multigraphs for any integers $\kappa \ge r \ge 3$.
\end{theorem}


In this paper, we show two results,
Theorems~\ref{thm:edge_coloring-regular-simple-non-planar}
 and
\ref{thm:edge_coloring-regular-simple-planar},
about
the counting complexity of edge colorings over (regular) simple graphs,
one 
for the planar case and the other for the general (i.e., not necessarily
planar) case.
%
%
%
As noted before, the only interesting cases are for $\kappa  \ge r \ge 3$,
and in the planar case for $\kappa  \ge r = 3, 4, 5$.
We handle both the planar and nonplanar cases
together, but we first prove it for the case $\kappa = r$ 
in Section~\ref{sec:kappa=r}
and then for the case $\kappa > r$ in Section~\ref{sec:kappa>r}.

We remark that when the underlying graph has a loop
then there is no proper edge colorings so the corresponding Holant value is $0$.
Hence we may assume that our reduction proof for 
Theorems~\ref{thm:edge_coloring-regular-simple-non-planar}
 and 
\ref{thm:edge_coloring-regular-simple-planar}
 starts
with  the problem of edge colorings on loopless multigraphs in 
Theorem~\ref{thm:edge_coloring}.

\section{Case \texorpdfstring{$\kappa = r$}{kappa = r}}\label{sec:kappa=r}

We start with the statement of  Theorem~\ref{thm:edge_coloring}
specialized with $\kappa = r \ge 3$.
%
%
We want to prove a corresponding statement
for simple graphs,
i.e., graphs without self-loops or parallel edges.
We first recall a well-known fact that
for $r > 5$, there are \emph{no} planar $r$-regular simple graphs.
This is a topological fact from Euler's formula; for completeness
we include a simple proof.
For any simple plane graph $G$, i.e., a planar graph with
a planar embedding,  let 
${\sf V}$, ${\sf E}$ and ${\sf F}$  be respectively the number of vertices,
edges and faces formed by $G$. Then Euler's formula says that
\[{\sf V} - {\sf E} + {\sf F} = 2.\]
If $G$ is an  $r$-regular simple graph,
then counting the ends of every edge once
we get $r {\sf V} = 2 {\sf E}$.
Similarly $r {\sf V} \ge 3 {\sf F}$, since at every vertex
we encounter $r$ incident faces, but each face is counted at least three times
on a simple graph $G$. Hence $(1 - \frac{r}{2} +  \frac{r}{3}) {\sf V} \ge 2$.
In particular $r \le 5$.

For $\kappa \in \{3,4,5\}$ we prove
the following for planar $\kappa$-regular simple graphs.

\begin{theorem} \label{thm:edge_coloring:k=r=3_simple}
 \#$\kappa$-\textsc{EdgeColoring} is $\SHARPP$-complete 
over planar $\kappa$-regular simple graphs, for $\kappa \in \{3,4,5\}$.
\end{theorem}

\begin{proof}
We reduce the problem of \#$\kappa$-\textsc{EdgeColoring}
over planar $\kappa$-regular 
 multigraphs to planar $\kappa$-regular simple graphs 
in polynomial time, for $\kappa \in \{3,4,5\}$.

Let $G$ be any planar $\kappa$-regular
 multigraph. As mentioned before, we may also assume $G$ has no self-loops.
%
We will replace
every edge of $G$ by a suitable gadget with two external 
edges to obtain a planar 
$\kappa$-regular simple
graph $G'$. 
The key property of the  gadget is as follows:
Every proper edge $\kappa$-coloring assigns the same color to
the two external edges $e_1$ and $e_2$, and such
a coloring exists.
Note that, by permuting the set $[\kappa]$,
 for any fixed color value
for  $e_1$ and $e_2$ there are the same number $c$ of extensions to a
proper edge $\kappa$-coloring of the gadget.
With this property, the number of proper edge $\kappa$-colorings
of $G'$ is $c^{\sf E}$ times that of $G$, where ${\sf E}$
denotes the number of edges of $G$.

For $\kappa =3$, consider the gadget $H_3$ in Figure~\ref{fig:H3}.
\begin{figure}
\centering
\subfloat[An edge in $G$]{
\begin{tikzpicture}[scale=0.4]
\node [external] (00) at (1, -2) {}; 
\node [internal, scale=0.6] (0) at (0, 3) {};
\node [internal, scale=0.6] (1) at (6, 3) {};
\draw  (0) to  (1);
\end{tikzpicture}
}
\qquad
\qquad
\subfloat[Gadget $H_3$]{
\begin{tikzpicture}[scale=0.4]
\node at (6.5, 3) {$e_5$};
\node at (3.5, 5) {$e_3$};
\node at (3.5, 1) {$e_4$};
\node at (0, 3.5) {$e_1$};
\node at (12, 3.5) {$e_2$};
\node at (8.7, 5) {$e_6$};
\node at (8.7, 1) {$e_7$};
\node [external] (00) at (0, -2) {}; 
\node [internal, scale=0.6] (0) at (-2, 3) {};
\node [internal, scale=0.6] (1) at (2, 3) {};
\node [internal, scale=0.6] (2) at (10, 3) {};
\node [internal, scale=0.6] (3) at (14, 3) {};
\node [internal, scale=0.6] (4) at (6, -1) {};
\node [internal, scale=0.6] (5) at (6, 7) {};
\draw  (0) to  (1);
\draw  (2) to  (3);
\draw  (4) to  (1);
\draw  (4) to  (2);
\draw  (4) to  (5);
\draw  (5) to  (1);
\draw  (5) to  (2);
\end{tikzpicture}
}
 \caption{\label{fig:H3}The gadget for 3-regular graphs}
\end{figure}
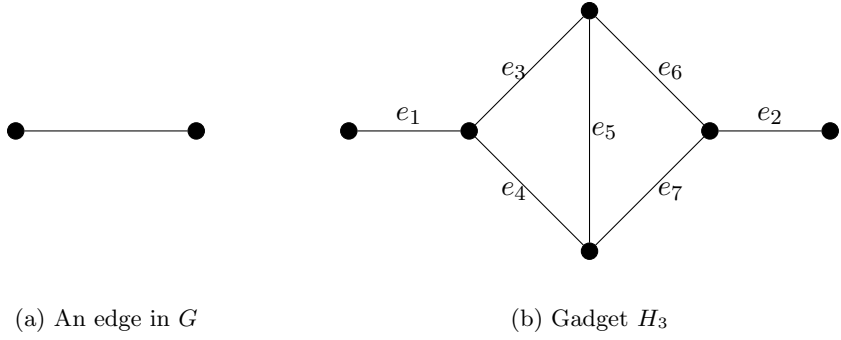
Replacing every edge of $G$ by $H_3$ we get a planar $3$-regular simple 
graph $G'$. 
To prove the key property for $H_3$, suppose $e_1$ is assigned a color
 {\sl Red},
then its two adjacent edges $e_3$ and $e_4$ must be assigned the other two
colors {\sl Blue} and {\sl Green} in some order.
Then the middle vertical edge $e_5$ is assigned {\sl Red}, and its remaining
two adjacent edges $e_6$ and $e_7$
 must be assigned {\sl Blue} and {\sl Green}  in a unique way (depending
on the colors of $e_3$ and $e_4$).
It follows that $e_2$ is also assigned  {\sl Red}, and we found 
exactly $c=2$ proper edge $3$-colorings.  

For $\kappa =4$, consider the gadget $H_4$ in Figure~\ref{fig:H4}.
\begin{figure}
\centering
\subfloat[An edge in $G$]{
\begin{tikzpicture}[scale=0.35]
\node [external] (00) at (1, -2) {}; 
\node [internal, scale=0.6] (0) at (0, 3) {};
\node [internal, scale=0.6] (1) at (6, 3) {};
\draw  (0) to  (1);
\end{tikzpicture}
}
\qquad
\qquad
\subfloat[Gadget $H_4$]{
\begin{tikzpicture}[scale=0.35]
\node at (10, 3.4) {$e$};
\node at (6, 5.2) {$e_5$};
\node at (8.5,5.2) {$e_3$};
\node at (10.2,5.2) {$e_4$};
\node at (12,5.2) {$e_6$};
\node at (2,3.4) {$e_1$};
\node at (18,3.4) {$e_2$};
\node [external] (00) at (1, -2) {}; 
\node [internal, scale=0.6] (0) at (0, 3) {};
\node [internal, scale=0.6] (1) at (4, 3) {};
\node [internal, scale=0.6] (2) at (8, 3) {};
\node [internal, scale=0.6] (3) at (12, 3) {};
\node [internal, scale=0.6] (4) at (16, 3) {};
\node [internal, scale=0.6] (5) at (20, 3) {};
\node [internal, scale=0.6] (6) at (10, -1) {};
\node [internal, scale=0.6] (7) at (10, 7) {};
\draw  (0) to  (1);
\draw  (1) to  (2);
\draw  (2) to  (3);
\draw  (3) to  (4);
\draw  (4) to  (5);
\draw  (6) to  (1);
\draw  (6) to  (2);
\draw  (6) to  (3);
\draw  (6) to  (4);
\draw  (7) to  (1);
\draw  (7) to  (2);
\draw  (7) to  (3);
\draw  (7) to  (4);
\end{tikzpicture}
}
 \caption{\label{fig:H4}The gadget for 4-regular graphs}
\end{figure}
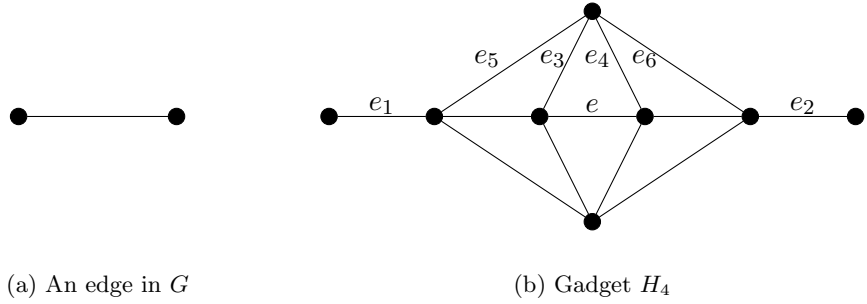
Replacing every edge of $G$ by $H_4$ we get a simple planar $4$-regular 
graph $G'$.
Suppose the middle horizontal edge $e$ is assigned {\sl Red}.
Then $e_3$ and $e_4$ are assigned two other
colors, say {\sl Blue} and {\sl Green}, in some order.
Let the 4th color be {\sl Yellow}. 
Then  $e_5$ and $e_6$ must be assigned  
  {\sl Yellow} and {\sl Red} in some order.
In particular exactly one of  $e_5$ or $e_6$ is assigned {\sl Yellow}.
It can be easily shown
 that the color ({\sl Blue} or {\sl Green})
 of the edge ($e_3$ or $e_4$) that is on the same internal
face as this edge colored {\sl Yellow}  ($e_5$ or $e_6$) 
 must be the color of both the external
edges $e_1$ and $e_2$, 
and once $e_5$ or  $e_6$ is chosen to be colored  {\sl Yellow}
 there is a unique extension.
 This implies that  the gadget $H_4$ satisfies the 
key property with the constant $c=12$: If the two 
external edges  $e_1$ and $e_2$
are to be {\sl Blue}, say, then  $e$ must be not {\sl Blue}, and one
of $e_3$ and $e_4$
must be {\sl Blue}, and the rest are all forced, with 
the unique choice of $e_5$ or  $e_6$ that shares a face with
the {\sl Blue} edge in $\{e_3, e_4\}$ colored  {\sl Yellow}.
 Thus there are
$3\times 2 \times 2 = 12$ choices.
%

We remark that the gadgets $H_3$ and $H_4$ are obtained from
breaking one edge in a tetrahedron and an octahedron, respectively. 

For $\kappa = 5$, we use a similar construction,
but we will argue slightly differently.
Let $H$ be any plane $5$-regular simple graph that can be
edge $5$-colored. Such graphs exist; for concreteness
we can take the icosahedron which is shown in Figure~\ref{fig:icosahedron}
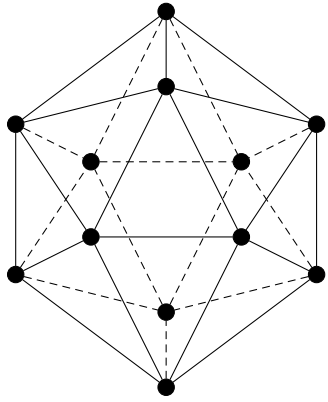
\begin{figure}
\centering
\begin{tikzpicture}[scale=0.5]
\node [external] (00) at (1, -2) {}; 
\node [internal, scale=0.6] (a1) at (4, 10) {};
\node [internal, scale=0.6] (a2) at (0, 7) {};
\node [internal, scale=0.6] (a3) at (0, 3) {};
\node [internal, scale=0.6] (a4) at (4, 0) {};
\node [internal, scale=0.6] (a5) at (8, 3) {};
\node [internal, scale=0.6] (a6) at (8, 7) {};
\node [internal, scale=0.6] (b1) at (4, 8) {};
\node [internal, scale=0.6] (b2) at (2, 4) {};
\node [internal, scale=0.6] (b3) at (6, 4) {};
\node [internal, scale=0.6] (c1) at (4, 2) {};
\node [internal, scale=0.6] (c2) at (2, 6) {};
\node [internal, scale=0.6] (c3) at (6, 6) {};
\draw [densely dashed] (a1) to (c2);
\draw [densely dashed] (a1) to (c3);
\draw[densely dashed] (a2) to (c2);
\draw[densely dashed] (a3) to (c1);
\draw[densely dashed] (a3) to (c2);
\draw[densely dashed] (a4) to (c1);
\draw[densely dashed] (a5) to (c1);
\draw[densely dashed] (a5) to (c3);
\draw[densely dashed] (a6) to (c3);
\draw[densely dashed] (c1) to (c2);
\draw[densely dashed] (c2) to (c3);
\draw[densely dashed] (c1) to (c3);
\draw  (a1) to (a2);
\draw  (a1) to (a6);
\draw  (a1) to (b1);
\draw  (a2) to (b1);
\draw  (a2) to (b2);
\draw  (a2) to (a3);
\draw  (a3) to (a4);
\draw  (a3) to (b2);
\draw  (a4) to (b2);
\draw  (a4) to (b3);
\draw  (a4) to (a5);
\draw  (a5) to (a6);
\draw  (a5) to (b3);
\draw  (a6) to (b1);
\draw  (a6) to (b3);
\draw  (b1) to (b3);
\draw  (b2) to (b3);
\draw  (b1) to (b2);
\end{tikzpicture}
 \caption{\label{fig:icosahedron}An icosahedron}
\end{figure} 
and a specific 5-coloring can be seen
in Figure~\ref{fig:icosahedron-colored}.
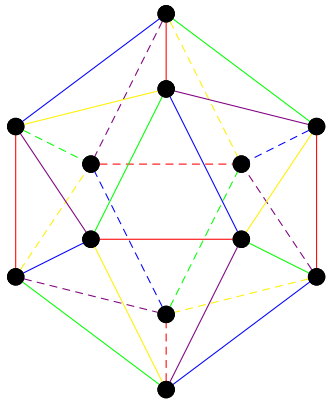
\begin{figure}[t]
\centering
\begin{tikzpicture}[scale=0.5]
\node [external] (00) at (1, -2) {}; 
\node [internal, scale=0.6] (a1) at (4, 10) {};
\node [internal, scale=0.6] (a2) at (0, 7) {};
\node [internal, scale=0.6] (a3) at (0, 3) {};
\node [internal, scale=0.6] (a4) at (4, 0) {};
\node [internal, scale=0.6] (a5) at (8, 3) {};
\node [internal, scale=0.6] (a6) at (8, 7) {};
\node [internal, scale=0.6] (b1) at (4, 8) {};
\node [internal, scale=0.6] (b2) at (2, 4) {};
\node [internal, scale=0.6] (b3) at (6, 4) {};
\node [internal, scale=0.6] (c1) at (4, 2) {};
\node [internal, scale=0.6] (c2) at (2, 6) {};
\node [internal, scale=0.6] (c3) at (6, 6) {};
\draw[violet] [densely dashed] (a1) to (c2);
\draw[yellow] [densely dashed] (a1) to (c3);
\draw[green]  [densely dashed] (a2) to (c2);
\draw[violet] [densely dashed] (a3) to (c1);
\draw[yellow] [densely dashed] (a3) to (c2);
\draw[red] [densely dashed] (a4) to (c1);
\draw[yellow] [densely dashed] (a5) to (c1);
\draw[violet] [densely dashed] (a5) to (c3);
\draw[blue] [densely dashed] (a6) to (c3);
\draw[blue] [densely dashed] (c1) to (c2);
\draw[red] [densely dashed] (c2) to (c3);
\draw[green] [densely dashed] (c1) to (c3);
\draw[blue]  (a1) to (a2);
\draw[green]  (a1) to (a6);
\draw[red]  (a1) to (b1);
\draw[yellow]  (a2) to (b1);
\draw[violet]  (a2) to (b2);
\draw[red]  (a2) to (a3);
\draw[green]  (a3) to (a4);
\draw[blue]  (a3) to (b2);
\draw[yellow]  (a4) to (b2);
\draw[violet]  (a4) to (b3);
\draw[blue]  (a4) to (a5);
\draw[red]  (a5) to (a6);
\draw[green]  (a5) to (b3);
\draw[violet]  (a6) to (b1);
\draw[yellow]  (a6) to (b3);
\draw[blue]  (b1) to (b3);
\draw[red]  (b2) to (b3);
\draw[green]  (b1) to (b2);
\node [internal, scale=0.6] (a1) at (4, 10) {};
\node [internal, scale=0.6] (a2) at (0, 7) {};
\node [internal, scale=0.6] (a3) at (0, 3) {};
\node [internal, scale=0.6] (a4) at (4, 0) {};
\node [internal, scale=0.6] (a5) at (8, 3) {};
\node [internal, scale=0.6] (a6) at (8, 7) {};
\node [internal, scale=0.6] (b1) at (4, 8) {};
\node [internal, scale=0.6] (b2) at (2, 4) {};
\node [internal, scale=0.6] (b3) at (6, 4) {};
\node [internal, scale=0.6] (c1) at (4, 2) {};
\node [internal, scale=0.6] (c2) at (2, 6) {};
\node [internal, scale=0.6] (c3) at (6, 6) {};
\end{tikzpicture}
\caption{\label{fig:icosahedron-colored}A possible proper edge $5$-coloring of the icosahedron}
\end{figure} 
Because $H$ is $5$-regular, $5 {\sf V} = 2 {\sf E}$,
so ${\sf V}$ is even. Now remove one edge of $H$ on the external face,
say between vertices $u$ and $v$.  This defines a graph $H'$.
Define $H_5$ to be the graph obtained from $H'$ by adding two
external edges $e_1$ and $e_2$
 from $u$ and $v$ (in the case of the icosahedron 
see Figures~\ref{fig:icosahedron-break-one-edge}
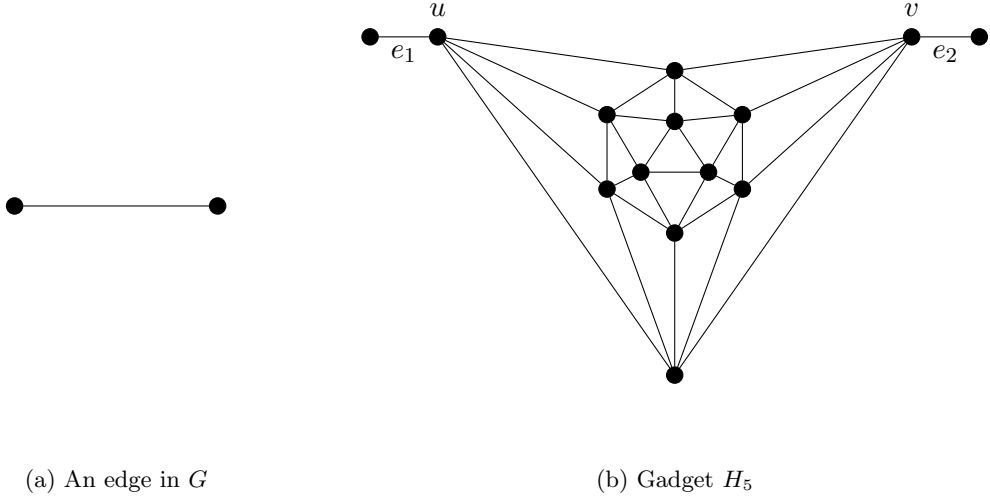
\begin{figure}
\centering
\subfloat[An edge in $G$]{
\begin{tikzpicture}[scale=0.45]
\node [external] (00) at (1, -2) {}; 
\node [internal, scale=0.6] (0) at (0, 5) {};
\node [internal, scale=0.6] (1) at (6, 5) {};
\draw  (0) to  (1);
\end{tikzpicture}
}
\qquad
\qquad
\subfloat[Gadget $H_5$]{
\begin{tikzpicture}[scale=0.45]
\node at (-1, 9.5) {$e_1$};
\node at (15, 9.5) {$e_2$};
\node [external] (00) at (0, -2) {}; 
\node [internal, scale=0.6] (a1) at (-2, 10) {};
\node [internal, scale=0.6] (a2) at (0, 10) {};
\node [internal, scale=0.6] (a3) at (14, 10) {};
\node at (0, 10.8) {$u$};
\node at (14, 10.8) {$v$};
\node [internal, scale=0.6] (a4) at (16, 10) {};
\node [internal, scale=0.6] (a5) at (7, 0) {};
\draw  (a1) to  (a2);
\draw  (a2) to  (a5);
\draw  (a5) to  (a3);
\draw  (a4) to  (a3);
\node [internal, scale=0.6] (b1) at (7, 9) {};
\node [internal, scale=0.6] (b2) at (5, 7.7) {};
\node [internal, scale=0.6] (b3) at (5, 5.5) {};
\node [internal, scale=0.6] (b4) at (7, 4.2) {};
\node [internal, scale=0.6] (b5) at (9, 7.7) {};
\node [internal, scale=0.6] (b6) at (9, 5.5) {};
\draw  (a2) to  (b1);
\draw  (a2) to  (b2);
\draw  (a2) to  (b3);
\draw  (a3) to  (b1);
\draw  (a3) to  (b5);
\draw  (a3) to  (b6);
\draw  (a5) to  (b3);
\draw  (a5) to  (b4);
\draw  (a5) to  (b6);
\node [internal, scale=0.6] (c1) at (7, 7.5) {};
\node [internal, scale=0.6] (c2) at (6, 6) {};
\node [internal, scale=0.6] (c3) at (8, 6) {};
\draw  (c1) to  (c2);
\draw  (c2) to  (c3);
\draw  (c3) to  (c1);
\draw  (b1) to  (b2);
\draw  (b1) to  (b5);
\draw  (b1) to  (c1);
\draw  (b2) to  (b3);
\draw  (b2) to  (c1);
\draw  (b2) to  (c2);
\draw  (b3) to  (b4);
\draw  (b3) to  (c2);
\draw  (b4) to  (c2);
\draw  (b4) to  (c3);
\draw  (b4) to  (b6);
\draw  (b6) to  (b5);
\draw  (b6) to  (c3);
\draw  (b5) to  (c1);
\draw  (b5) to  (c3);
\end{tikzpicture}
}
 \caption{\label{fig:icosahedron-break-one-edge}The gadget for 5-regular graphs}
\end{figure} 
and \ref{fig:icosahedron-gadget-colored}). 
\begin{figure}[t]
\centering
\begin{tikzpicture}[scale=0.45]
\tikzstyle{s1} = [red]
\tikzstyle{s2} = [green]
\tikzstyle{s3} = [blue]
\tikzstyle{s4} = [yellow]
\tikzstyle{s5} = [violet]
\node at (-1, 9.5) {$e_1$};
\node at (15, 9.5) {$e_2$};
\node [external] (00) at (0, -2) {}; 
\node [internal, scale=0.6] (a1) at (-2, 10) {};
\node [internal, scale=0.6] (a2) at (0, 10) {};
\node [internal, scale=0.6] (a3) at (14, 10) {};
\node at (0, 10.8) {$u$};
\node at (14, 10.8) {$v$};
\node [internal, scale=0.6] (a4) at (16, 10) {};
\node [internal, scale=0.6] (a5) at (7, 0) {};
\draw[s1]  (a1) to  (a2);
\draw[s3]  (a2) to  (a5);
\draw[s2]  (a5) to  (a3);
\draw[s1]  (a4) to  (a3);
\node [internal, scale=0.6] (b1) at (7, 9) {};
\node [internal, scale=0.6] (b2) at (5, 7.7) {};
\node [internal, scale=0.6] (b3) at (5, 5.5) {};
\node [internal, scale=0.6] (b4) at (7, 4.2) {};
\node [internal, scale=0.6] (b5) at (9, 7.7) {};
\node [internal, scale=0.6] (b6) at (9, 5.5) {};
\draw[s5]  (a2) to  (b1);
\draw[s2]  (a2) to  (b2);
\draw[s4]  (a2) to  (b3);
\draw[s4]  (a3) to  (b1);
\draw[s3]  (a3) to  (b5);
\draw[s5]  (a3) to  (b6);
\draw[s5]  (a5) to  (b3);
\draw[s1]  (a5) to  (b4);
\draw[s4]  (a5) to  (b6);
\node [internal, scale=0.6] (c1) at (7, 7.5) {};
\node [internal, scale=0.6] (c2) at (6, 6) {};
\node [internal, scale=0.6] (c3) at (8, 6) {};
\draw[s2]  (c1) to  (c2);
\draw[s1]  (c2) to  (c3);
\draw[s3]  (c3) to  (c1);
\draw[s3]  (b1) to  (b2);
\draw[s2]  (b1) to  (b5);
\draw[s1]  (b1) to  (c1);
\draw[s1]  (b2) to  (b3);
\draw[s4]  (b2) to  (c1);
\draw[s5]  (b2) to  (c2);
\draw[s2]  (b3) to  (b4);
\draw[s3]  (b3) to  (c2);
\draw[s4]  (b4) to  (c2);
\draw[s5]  (b4) to  (c3);
\draw[s3]  (b4) to  (b6);
\draw[s1]  (b6) to  (b5);
\draw[s2]  (b6) to  (c3);
\draw[s5]  (b5) to  (c1);
\draw[s4]  (b5) to  (c3);
\node [internal, scale=0.6] (a1) at (-2, 10) {};
\node [internal, scale=0.6] (a2) at (0, 10) {};
\node [internal, scale=0.6] (a3) at (14, 10) {};
\node [internal, scale=0.6] (a4) at (16, 10) {};
\node [internal, scale=0.6] (a5) at (7, 0) {};
\node [internal, scale=0.6] (b1) at (7, 9) {};
\node [internal, scale=0.6] (b2) at (5, 7.7) {};
\node [internal, scale=0.6] (b3) at (5, 5.5) {};
\node [internal, scale=0.6] (b4) at (7, 4.2) {};
\node [internal, scale=0.6] (b5) at (9, 7.7) {};
\node [internal, scale=0.6] (b6) at (9, 5.5) {};
\node [internal, scale=0.6] (c1) at (7, 7.5) {};
\node [internal, scale=0.6] (c2) at (6, 6) {};
\node [internal, scale=0.6] (c3) at (8, 6) {};
\end{tikzpicture}
\caption{\label{fig:icosahedron-gadget-colored}A proper edge $5$-coloring of the icosahedron gadget where $e_1, e_2$ have the same color}
\end{figure}
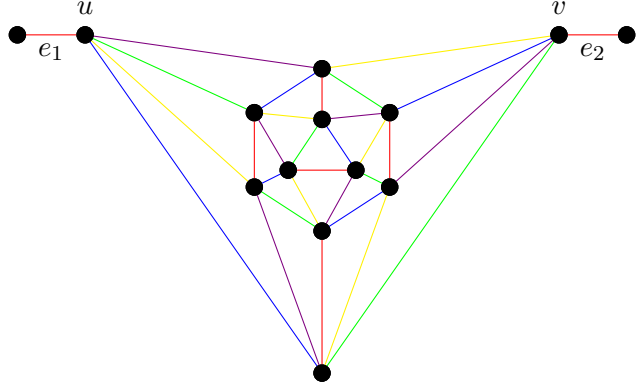 
We claim that $H_5$ satisfies the
key property.
To prove that we only need to show that 
any edge $5$-coloring of $H_5$ must use the same color on the
two external edges.
For a contradiction, suppose this is not so. Then
for some $5$-coloring of $H_5$,
there are two distinct colors {\sl Red} and {\sl Green}, say,
such that  the two external edges are colored {\sl Red} and {\sl Green},
respectively.
Consider the set $S$ of edges in $H_5$ colored {\sl Red} or {\sl Green}.
Clearly  the two external vetices
are incident to edges in $S$.
Since all other vertices of  $H_5$ have degree $5$, all
vertices of  $H_5$ are also  incident to edges in $S$.
$S$ must be a vertex disjoint union of (zero or more) alternating
{\sl Red}-{\sl Green} cycles, and exactly 
one alternating {\sl Red}-{\sl Green} path
starting and ending with the two external edges. 
This alternating path starts and ends
with distinct colors, and 
 therefore it has an even number of
edges and thus an odd number of vertices. The alternating cycles have
an  even number of vertices. This is a contradiction to
${\sf V} \equiv 0 \bmod 2$.
(The exact value  $c >0$ for $H_5$ is a constant and is unimportant
to the existence of this reduction.)
\end{proof}



The idea of the proof can also be adapted to prove the following theorem
without the planarity resquirement.

\begin{theorem}\label{thm:edge_coloring:non-planar-k=r>=3_simple}
\#$\kappa$-\textsc{EdgeColoring} is $\SHARPP$-complete
over $\kappa$-regular simple graphs, for all $\kappa  \ge 3$.
\end{theorem}

\begin{proof}
Let $H$ be a $\kappa$-regular simple graph that admits a proper edge 
$\kappa$-coloring. For now assume such a graph exists.
If we consider any color, the edges with that color form a perfect matching
of $H$, because the number of colors  and the regularity parameter
are both $\kappa$. Hence $H$ has an even number of vertices.

Take any two adjacent vertices $u$ and $v$.
Now remove one edge $(u,v)$  of $H$; this defines a graph $H'$.
Define $H^*$ to be the graph obtained from $H'$ by adding two
external edges $e_1$ and $e_2$
incident to $u$ and $v$, respectively.
$H^*$ will be the gadget replacing every edge in a 
$\kappa$-regular multigraph  $G$
to produce a $\kappa$-regular  simple graph  $G'$.
By the same proof above we see that in any
edge $\kappa$-coloring of $H'$, the  $\kappa - 1$ incident  edges at $u$ 
receive the same set of  $\kappa - 1$ colors as the $\kappa - 1$ incident  edges at $v$.
Thus  $H^*$ satisfies the key property stated in the proof of 
Theorem~\ref{thm:edge_coloring:k=r=3_simple}.
This completes the proof, provided we exhibit one such graph $H$.

Let $n = \kappa!$. Since $\kappa  \ge 3$, we have $n/2 \ge \kappa$.
We take the vertex set of $H$ to be the additive group $\mathbb{Z}_n$.
We define $\kappa$ edge disjoint perfect matchings, and the edge set of
$H$ as their union.
For $1 \le \ell \le  \lfloor \kappa/2 \rfloor$,
let $E_\ell = \{(i, i+\ell) \mid 0 \le i < \ell\}$,
and $E'_\ell =  \{(i+ \ell, i+ 2 \ell) \mid 0 \le i < \ell \}$.
Let $M_\ell$ and $M'_\ell$ be respectively
 the orbits of $E_\ell$ and  $E'_\ell$ by the subgroup 
$(2\ell) \mathbb{Z}_n \simeq \mathbb{Z}_{n/(2\ell)}$.
Thus
\begin{eqnarray*}
M_\ell &=& \bigcup_{i=0}^{\ell-1} \{(i, i+ \ell), (i+ 2\ell, i+ 3\ell), \ldots, (i+ n-2\ell, i+ n-\ell)\},\\
M'_\ell &=& \bigcup_{i=0}^{\ell-1} \{(i+ \ell, i+ 2\ell), (i+ 3\ell, i+ 4\ell), \ldots, (i+ n-\ell, i+ n)\}.
\end{eqnarray*}
Note that $2\ell \le \kappa$ and so $2\ell \mid n$. 
If $\kappa$ is odd, then add one more perfect matching
$M = \{(j, j+n/2) \mid 0 \le j < n/2\}$.
This defines  $\kappa$ edge disjoint perfect matchings.
\end{proof}




\section{Case \texorpdfstring{$\kappa > r$}{kappa > r}}\label{sec:kappa>r}

In this section, we prove similar results for simple graphs where $\kappa > r$.
We use $J_\kappa$
to denote the all-1 $\kappa$ by $\kappa$ matrix.
We call a gadget $r$-regular if its internal vertices all have degree $r$.

\begin{theorem}\label{thm:edge_coloring_simple_planar:k>r}
\#$\kappa$-\textsc{EdgeColoring} is $\SHARPP$-complete over planar $r$-regular simple graphs for all $\kappa > r$, and $r \in \{ 3, 4, 5 \}$.
\end{theorem}
\begin{proof}
We reduce the problem \#$\kappa$-\textsc{EdgeColoring} over planar $r$-regular multigraphs,
which is \SHARPP-complete by Theorem \ref{thm:edge_coloring}, to planar
 $r$-regular simple graphs in polynomial time.
Then the statement of the theorem follows.

For now let us suppose
 we can construct a binary planar $r$-regular simple gadget $f$
consisting of only $\AD_{r, \kappa}$ signatures.
We say a function is  domain invariant if it is invariant
under any permutation of the input domain elements.
Clearly $\AD_{r, \kappa}$ is domain invariant,  and so
$f$ is also domain invariant.
Hence its  $\kappa$ by $\kappa$ signature matrix
has the form $A = (a - b) I_\kappa + b J_\kappa$
for some nonnegative integers $a$ and $b$, i.e., it has values $a$ 
and $b$ on and off the main diagonal, correspondingly.
In particular, $A$ is symmetric.
For any $n \ge 1$, consider a binary gadget $f_n$
obtained by connecting $f$ sequentially $n$ times,
with $f_1 = f$
(see Figure~\ref{fig:gadget:k>=r:binary:interpolation}).
Note that the order in which the dangling edges of
 $f$ are connected to each other
makes no difference in this construction
because  $A$ is symmetric.
Obviously,  each gadget $f_n$ for $n \ge 1$ is still
a planar $r$-regular simple gadget, with
a domain invariant signature matrix $A^n$.
Assume further that we can satisfy  $a \ne b$.

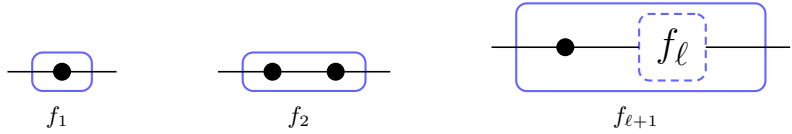
\begin{figure}[t]
 \centering
 \captionsetup[subfigure]{labelformat=empty}
\subfloat[$f_1$]{
  \begin{tikzpicture}[scale=\scale,transform shape,node distance=\nodeDist,semithick]
   \node[external] (0)              {};
   \node[internal] (1) [right of=0] {};
   \node[external] (2) [right of=1] {};
   \path (0) edge (1)
         (1) edge (2);
   \begin{pgfonlayer}{background}
    \node[draw=\borderColor,thick,rounded corners,inner xsep=8pt,inner ysep=4pt,fit = (1), transform shape=false] {};
   \end{pgfonlayer}
  \end{tikzpicture}}
 \qquad
\subfloat[$f_2$]{
  \begin{tikzpicture}[scale=\scale,transform shape,node distance=\nodeDist,semithick]
   \node[external] (0)              {};
   \node[internal] (1) [right of=0] {};
   \node[internal] (2) [right of=1] {};
   \node[external] (3) [right of=2] {};
   \path (0) edge (1)
         (1) edge (2)
         (2) edge (3);
   \begin{pgfonlayer}{background}
    \node[draw=\borderColor,thick,rounded corners,inner xsep=8pt,inner ysep=4pt,fit = (1) (2), transform shape=false] {};
   \end{pgfonlayer}
  \end{tikzpicture}}
 \qquad
\subfloat[$f_{\ell+1}$]{
  \begin{tikzpicture}[scale=\scale,transform shape,node distance=\nodeDist,semithick]
   \node[external] (0)              {};
   \node[external] (1) [right of=0] {};
   \node[external] (2) [right of=1] {};
   \node[external] (3) [right of=2] {\Huge $f_\ell$};
   \node[external] (4) [right of=3] {};
   \node[external] (5) [right of=4] {};
   \path (0) edge node[internal] (e1) {} (3)
         (3) edge (5);
   \begin{pgfonlayer}{background}
    \node[draw=\borderColor,thick,densely dashed,inner sep=0pt, rounded corners,fit = (3), transform shape=false] {};
    \node[draw=\borderColor,thick,rounded corners,inner xsep=8pt,inner ysep=4pt,fit = (1) (3) (4), transform shape=false] {};
   \end{pgfonlayer}
  \end{tikzpicture}}
\caption{Recursive construction to interpolate the binary \textsc{Equality} signature $=_2$.
All vertices are assigned the gadget $f$. The dangling edges of $f$ can be connected in any order.}
 \label{fig:gadget:k>=r:binary:interpolation}
\end{figure}

Let $\Omega = \Omega(G)$ be an input instance to $\PlHolant(\AD_{r, \kappa})$,
 where $G = (V, E)$ is a planar $r$-regular multigraph.
Here every vertex in $V$ is assigned a symmetric signature $\AD_{r, \kappa}$.
Let $F \subseteq E$ be the subset of the edges of $G$
each of which is parallel to at least one other edge.
For any $n \ge 1$, we construct $\Omega_n = \Omega(G_n)$ to be the
 instance of $\PlHolant(\AD_{r, \kappa})$ 
obtained by replacing every edge in $F$ by the gadget $f_n$.
Clearly, the underlying graph $G_n$ of $\Omega_n$
 is simple  planar  $r$-regular
and has size polynomial in $n$ and the size of $\Omega$.
Let $m = |F|$, which is also the number of the gadgets $f_n$ used
 in $\Omega_n$.
Let $\Omega' = \Omega(G')$ be a signature grid obtained from $\Omega$
by placing the binary \textsc{Equality} signature $=_2$ in the middle of every edge in $F$. 
Clearly, $\Holant_{\Omega'} = \Holant_{\Omega}$.
We show how to compute this value from
the values $\Holant_{\Omega_n}$ where $n \ge 1$ in polynomial time.


The matrix $A = (a - b) I_\kappa + b J_\kappa$ has eigenvalues
 $\lambda_1 = a - b + \kappa b$ and $\lambda_i = a - b$ 
for $2 \le i \le \kappa$.
Since it is real symmetric,
 it can be orthogonally diagonalized over $\mathbb{R}$, i.e.,
there exist a real orthogonal matrix $S$
and a real diagonal matrix $D = (\lambda_i)_{i = 1}^\kappa$ such that 
$A = S^T D S$.
As $a \ne b$ we have $\lambda_i = a - b \ne 0$ for $2 \le i \le \kappa$.
Also $\lambda_1 = a + (\kappa-1)b >0$. 
It follows that the matrices $A$ and $D$ are nondegenerate.
Clearly $A^n = S^T D^n S$,  for all $n \ge 1$,
and is also symmetric.

We can write the $(i,j)$ entry
$(A^n)_{i j} = \sum_{\ell = 1}^\kappa \alpha_{i j \ell} \lambda_\ell^n$
by a formal expansion,
for every $n \ge 1$ and some real $\alpha_{i j \ell}$'s
that are dependent on $S$, but independent of $n$ and $\lambda_\ell$,
where $1 \le i, j, \ell \le \kappa$. 
By the formal expansion of the symmetric matrix $A^n$ above,
we have $\alpha_{i j \ell} = \alpha_{j i \ell}$.
%

In the evaluation of $\Holant_{\Omega_n}$,
we stratify the (edge) assignments in $\Omega_n$
based on the colors assigned to the dangling edges of
the binary gadgets $f_n$ in $\Omega_n$ as follows.
Denote by $\tau = (t_{i j})_{1 \le i \le j \le \kappa}$
a nonnegative integer tuple with entries indexed by ordered pairs of numbers
and that satisfy $\sum_{1 \le i \le j \le \kappa} t_{i j} = m$.
Let $c_\tau$ 
be the sum over all assignments of the products of all signatures
in $\Holant_{\Omega_n}$,
 except the contributions by the gadgets $f_n$
such that the endpoints of precisely $t_{i j}$ constituent gadgets $f_n$
receive the assignments $(i, j)$ (in either order of the end points)
 for every $1 \le i \le j \le \kappa$.
Let $\mathcal T$ denote the set of all such possible tuples $\tau$, where
 $|\mathcal T| = \binom{m + \kappa (\kappa + 1) / 2 - 1}{\kappa (\kappa + 1) / 2 - 1}$.
Then
\[
\Holant_{\Omega_n} = \sum_{\tau \in \mathcal T} c_\tau \prod_{1 \le i \le j \le \kappa} ((A^n)_{i j})^{t_{i j}}
= \sum_{\tau \in \mathcal T} c_\tau \prod_{1 \le i \le j \le \kappa} (\sum_{\ell = 1}^\kappa \alpha_{i j \ell} \lambda_\ell^n)^{t_{i j}}. 
\]
Expanding out the last sum and rearranging the terms we get
\begin{equation}\label{eq:interpolation-lin-sys-vand} 
\Holant_{\Omega_n}
= \sum_{\substack{i_1 + \ldots + i_\kappa = m \\ i_1, \ldots, i_\kappa \ge 0}} b_{i_1, \ldots, i_\kappa} ( \prod_{\ell = 1}^\kappa \lambda_\ell^{i_\ell} )^n,
\end{equation}
for some numbers $b_{i_1, \ldots, i_\kappa}$ independent of
$\lambda_\ell$'s.

This can be viewed as a linear system with the unknowns $b_{i_1, \ldots, i_\kappa}$ with the rows indexed by $n$. 
The number of unknowns is  $\binom{m + \kappa - 1}{\kappa - 1}$
which is polynomial in $n$ and the size of the input instance $\Omega$, since $\kappa$ is a constant.
The values $\prod_{\ell = 1}^\kappa \lambda_\ell^{i_\ell}$
can all be computed in polynomial time.

On the other hand, it is clear that
\[
\Holant_{\Omega'} = \displaystyle \sum_{\substack{i_1 + \ldots + i_\kappa = m \\ i_1, \ldots, i_\kappa \ge 0}} b_{i_1, \ldots, i_\kappa}.
\]
We show next how to compute the value $\Holant_{\Omega'}$
from the values $\Holant_{\Omega_n}$ where $n \ge 1$ in polynomial time.
The coefficient matrix of the linear system (\ref{eq:interpolation-lin-sys-vand})
 is Vandermonde.
However, when $m \ge 1$, it is not of full rank
because the coefficients $\prod_{\ell = 1}^\kappa \lambda_\ell^{i_\ell}$
are not pairwise distinct, and therefore
it has repeating columns.
Nevertheless, when there are two repeating columns we replace
the corresponding unknowns $b_{i_1, \ldots, i_\kappa}$ and $b_{i'_1, \ldots, i'_\kappa}$
with their sum as a new variable; we repeat this replacement procedure until
there are no repeating columns.
Since all $\lambda_\ell \ne 0$, for $1 \le \ell \le \kappa$,
after the replacement,
 we have a Vandermonde system of full rank.
Therefore we can solve this modified linear system
in polynomial time and find the desired value $\Holant(\Omega) = \Holant(\Omega') = \displaystyle \sum_{\substack{i_1 + \ldots + i_\kappa = m \\ i_1, \ldots, i_\kappa \ge 0}} b_{i_1, \ldots, i_\kappa}$.

\begin{remark}
The above proof did not use the fact in this specific case
$\lambda_2 = \ldots = \lambda_\kappa$.
The general Vandermonde argument does not need this property.
\end{remark}

We are left to prove that we can construct
a binary planar $r$-regular simple gadget $f$ with distinct
diagonal and off diagonal values  $a \ne b$ in its signature matrix $A$.
From the previous section we know that for every $r \in \{ 3, 4, 5 \}$
there exists a binary connected planar $r$-regular simple gadget.
For example, one can take $f$ to be the gadget $H_r$ described in Theorem~\ref{thm:edge_coloring:k=r=3_simple},
i.e., the ones shown in Figures~\ref{fig:H3}, \ref{fig:H4} and \ref{fig:icosahedron-break-one-edge}, correspondingly. 
The connectedness of each $H_3, H_4$ and $H_5$ is self-evident.
In the proof of Theorem~\ref{thm:edge_coloring:k=r=3_simple}, 
we also showed that
in each  respective case $r \in \{ 3, 4, 5 \}$,
$H_r$  can be $r$-colored  so that
the two dangling edges are colored with the same color.
Fix such a coloring.
Since now we are given $\kappa > r$, 
we can modify this coloring to change the color of one of the dangling edges
to one of the remaining $\kappa - r \ge 1$  colors
different from the $r$ colors used.
This still produces a proper $\kappa$-edge coloring
but in which the dangling edges are now colored differently.
Therefore $b \neq 0$.
If $a \neq b$, then we are done.

Suppose $a = b$, so we can write $A = b J_\kappa$.
$J_\kappa$ can be written as the column vector $(1, 1, \ldots, 1)_\kappa^T$
times the row vector $(1, 1, \ldots, 1)_\kappa$.
Clearly, in terms of evaluating signature value,
 replacing an internal edge within
the gadget $f$  with another copy of the gadget $f$ itself
is equivalent, up to a nonzero scalar $b$,
 to cutting this edge in half and assigning 
the unary signature $(1, 1, \ldots, 1)_\kappa$
to the two new degree one vertices.
Now choose a path $P$ in the gadget $f = H_r$
starting in one dangling edge and ending in the other.
Since $f$ is connected, such a path exists.
Suppose it  has  $s$ internal nodes,
where  $s \ge 1$.
Let $g$ be the binary gadget obtained
by replacing every edge in $f$ that is \emph{not}
 in $P$ by a  new copy of $f$
(by the symmetry of
$f$,
the order of the dangling edges of $f$ in the replacement bears no difference).
It is easy to see that $g$ is a binary connected planar $r$-regular simple gadget. 
Note that connecting $(1, 1, \ldots, 1)_\kappa$ to $\AD_{n, \kappa}$
gives us the signature $(\kappa - n + 1) \AD_{n - 1, \kappa}$ where $1 \le n \le \kappa$
so it is the signature $\AD_{n - 1, \kappa}$,
up to the nonzero scalar $\kappa - n + 1$.
Thus $g$, as a constraint function,
is equivalent to simply a path $P$ with $\AD_{2, \kappa} \equiv (\neq_2)$
assigned to the intermediate vertices.
Hence, (up to an easily computable nonzero constant)
its signature matrix is $(J_\kappa - I_\kappa)^{s}$.
It is easy to verify that $(J_\kappa - I_\kappa)^{s}$ has
the form $(-1)^s I_\kappa + c_{s, \kappa} J_\kappa$ for some
$c_{s, \kappa}$, and $(-1)^s$ is the difference of the diagonal value
and the off diagonal value, in particular nonzero.
Hence all the necessary conditions on $g$ are satisfied and 
we can use $g$ instead of $f$.
\end{proof}

Finally, we deal with the nonplanar case for $\kappa > r \ge 3$.
\begin{theorem}\label{thm:edge_coloring_simple_nonplanar:k>r}
\#$\kappa$-\textsc{EdgeColoring} is $\SHARPP$-complete over
$r$-regular simple graphs for $\kappa > r \ge 3$.
\end{theorem}
\begin{proof}
As in the proof of Theorem~\ref{thm:edge_coloring_simple_planar:k>r},
we reduce the problem \#$\kappa$-\textsc{EdgeColoring} 
over planar $r$-regular multigraphs, 
which is \SHARPP-complete by Theorem~\ref{thm:edge_coloring},
to (non-necessarily planar) $r$-regular simple graphs in polynomial time.
Then the statement of the theorem follows.

Following the steps of the proof of Theorem~\ref{thm:edge_coloring_simple_planar:k>r},
we see that it suffices to produce a binary connected $r$-regular simple gadget $f$
that is $r$-colorable with the additional condition
that the dangling edges are colored with the same color.
The only difference is that we do not
 require $f$ to be a planar gadget (and for $r>5$ this would be impossible).

Let $U = \{ u_1, u_2, \ldots, u_r \}, V = \{ v_1, v_2, \ldots, v_r \}$ and $\{ u, v \}$
be pairwise disjoint sets of vertices.
We take the vertex set of $f$ to be $U \cup V \cup \{ u, v \}$
where $u, v$ will be the external endpoints of the two dangling edges.
Next, we define the edge set of $f$ to be $A \cup B \cup C \cup D$ where
\begin{eqnarray*}
A &=& \{ (u_r, u_i) \mid 1 \le i \le r - 1 \},\\
B &=& \{ (v_r, v_i) \mid 1 \le i \le r - 1 \},\\
C &=& \{ (u_i, v_j) \mid 1 \le i, j \le r - 1 \}, \\
D &=&  \{ (u, u_r), (v_r, v) \}.
\end{eqnarray*}
Here $(u, u_r)$ and $(v_r, v)$ are the dangling edges of $f$;
there is a complete bipartite graph $K_{r-1, r-1}$
between $\{ u_1, u_2, \ldots, u_{r-1}\}$ and $ \{ v_1, v_2, \ldots, v_{r-1}\}$,
while $u_r$ is connected to all of the former, and
$v_r$ is  connected to all of the latter.
It is easy to see that $f$ is indeed a binary connected $r$-regular simple gadget.

We will exhibit a proper $r$-edge coloring of $f$.
For this, we first label the $r$ colors by
the elements of the additive group $\mathbb Z_r = \{ 0, 1, \ldots, r - 1 \}$.
Now define the edge coloring of $f$ as follows:
each $(u_i, v_j) \in B$ is colored $i + j \bmod r$ for 
$1 \le i, j \le r - 1$,
each $(u_r, u_j) \in A$ and each $(v_j, v_r) \in C$ is colored $j$,
for $1 \le j \le r-1$,
and $(u, u_r), (v_r, v) \in D$ are colored $0$.
It is easy to see that this coloring satisfies
all the necessary requirements.
\end{proof}

\bibliography{references}

\end{document}